\def\year{2020}\relax
\documentclass[letterpaper]{article} 
\usepackage{aaai20}  
\usepackage{times}  
\usepackage{helvet} 
\usepackage{courier}  
\usepackage[hyphens]{url}  
\usepackage{graphicx} 
\usepackage{graphicx}  
\nocopyright
\title{Collaborative Machine Learning Markets
with Data-Replication-Robust Payments}
\author{\Large \textbf{\makebox[0.22\linewidth]{Olga Ohrimenko} \makebox[0.22\linewidth]{Shruti Tople} \makebox[0.22\linewidth]{Sebastian Tschiatschek}}\\
\Large \textbf{}\\ 
\large{Microsoft Research}\\ 
\{oohrim,t-shtopl,setschia\}@microsoft.com 
}

\usepackage{amsmath,amssymb,amsthm}
\usepackage{paralist}
\usepackage{wrapfig}
\usepackage{lineno}
\usepackage{thmtools}
\usepackage{subcaption}
\usepackage{xcolor}

\newif\ifFull
\newif\ifShort

\Fullfalse

\ifFull
\Shortfalse
\else
\Shorttrue
\fi

\newcommand{\model}{\ensuremath{\mathcal{M}}}
\newcommand{\trdata}{\ensuremath{\mathcal{X}}}
\newcommand{\valdata}{\ensuremath{\mathcal{V}}}
\newcommand{\cat}{\ensuremath{\oplus}}
\newcommand{\gain}{\ensuremath{\mathcal{G}}} 

\newcommand{\shap}{\psi}
\newcommand{\shapR}{\shap^R}
\newcommand{\shapRold}{\shap^R_{\mathsf{old}}}
\newcommand{\shapRnew}{\shap^R_{\mathsf{new}}}
\newcommand{\fee}{A}
\newcommand{\price}{a}
\newcommand{\vsing}{v}
\newcommand{\vmult}{w}
\newcommand{\parties}{\mathbf{M}}
\newcommand{\partiesR}{\mathbf{M}^R}
\newcommand{\partiesD}{\mathbf{D}}

\newtheorem{definition}{Definition}

\newtheorem{claim}{Claim}
\newtheorem{condition}{Condition}
\newtheorem{lemma}{Lemma}

\newtheorem{corollary}{Corollary}

\begin{document}

\maketitle

\begin{abstract}
  We study the problem of collaborative machine learning markets where multiple parties can achieve improved performance on their machine learning tasks
  by combining their training data.
  We discuss desired properties for these machine learning markets in terms of fair revenue distribution and potential threats, including data replication.
  We then instantiate a collaborative market for cases where parties share a common machine learning task and where parties' tasks are different. Our marketplace incentivizes
parties to submit high quality training and true validation data.
To this end, we introduce a novel payment division function that is robust-to-replication and customized output models that perform well only on requested machine learning tasks.
  In experiments, we validate the assumptions underlying our theoretical analysis and show that these are approximately satisfied for commonly used machine learning models.
\end{abstract}




\section{Introduction}

One of the main obstacles for training well-performing machine learning models  is the limited availability of sufficiently diverse labeled training data.
However, the data needed to train good models often exists but is not easy to leverage as it is distributed and owned by multiple parties.
For instance, in the medical domain, important data about patients that could be used for learning diagnostic support systems for cancer might be in possession of different hospitals, each of which holding different data, (e.g., from a specific geographical region with different demographics).
Typically, by pooling the available data, the hospitals could train better machine learning models for their application than they could using only their own data.
As all hospitals would benefit from a better machine learning model obtained through data sharing, there is a need for collaborative machine learning.

Naturally, this type of collaboration raises questions in terms of how to incentivize parties to participate in such a collaborative machine learning effort.
Unfortunately,  in many applications there is a clear and natural incentive for parties to provide quality data or share their data to begin with.
Financially rewarding parties to incentivize participation seems to be natural in such cases but has to be done with great care.
For example, a fixed price per data point could motivate parties to gather large amounts of low quality or fake data if there is no mechanism to control data quality.
Another reason that may dis-incentivize parties from sharing data could stem
from privacy and integrity concerns regarding the use of party's data once it is shared.
For example, once a party's data is released it can be easily reused or  sold.

\begin{figure}[!tbp]
	\centering
	\includegraphics[width=1.\linewidth]{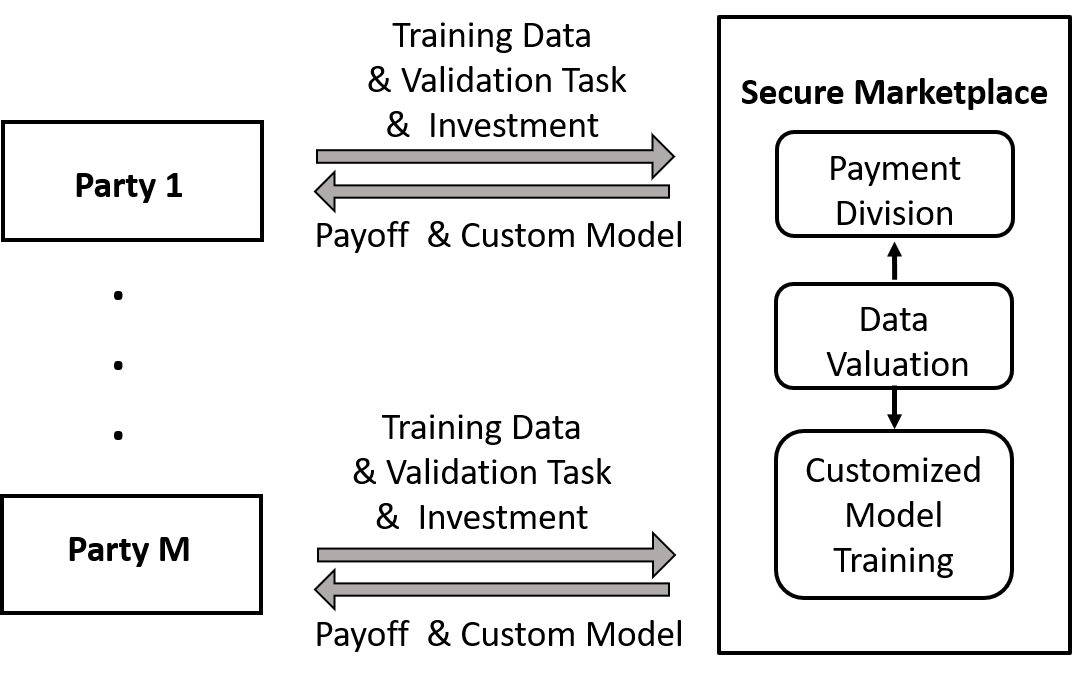}
	\caption{Collaborative marketplace setup.}
	\label{fig:setting}
\end{figure}

To overcome these challenges and enable collaborative machine learning in the outlined setting,
there is a need for a secure machine learning marketplace for joint model training that guarantees fair incentives for participation and ensures secure data handling.
In this work, we propose a cloud-based collaborative machine learning platform accessible to parties for submitting data and machine learning tasks.
Submitted training data represents the data that a party is willing to contribute/sell, while submitting the validation data can be seen as a specification of the machine learning model a party is willing to buy.
After a \emph{trade} in this market, a participating party obtains a model trained on the data available to the market and customized for its task.
Crucially, there is no sharing of data and parties are only provided this customized model (or query interface).
As a result, only information relevant to the validation task
is released through the model, limiting the possibility of copying and reusing the data for other tasks.
Such markets allow multiple parties to jointly train machine learning models based on the training data provided by all of the parties and achieve improved performance on their own tasks.
Parties pay \emph{to} the market for the improvement on their validation tasks and get paid \emph{by} the market for the contribution of their training data to the tasks of others.
The market can support a single validation task scenario, for example, where hospitals bring together their data to train a single model
for detecting cancer. 
Furthermore, it also supports scenarios where one's data can contribute to multiple tasks.
An overview of our envisioned marketplace is shown in Figure~\ref{fig:setting}.

Our market is enabled by three main components:
\begin{inparaenum}
  \item \emph{Data valuation.}
Each model in the market is trained on the data that is best suited for
the specified validation task.
Similar to recent and concurrent work~\cite{jia2019towards,ghorbani2019} we can use Shapley values~\cite{shapley1953value}, a solution concept from game theory discussed later, to match data to tasks.
  \item \emph{Customized model training.} We ensure that the model that a party receives is only suited for its own task but not the other tasks available on the market.
As a result each party is incentivized to provide its true validation task.
  \item \emph{Payment division.} Each party receives a reward proportional to how useful its data is for training
other models. We again use Shapley values, however, in this case to determine fair payoffs.
\end{inparaenum}

 One of the key challenges in designing such data marketplaces comes from the very nature of data, i.e., free replication.
 This means that the reward of a party submitting copies of the same dataset should not be more than the party submitting it once.
 Indeed the authors of~\cite{agarwal2019marketplace}
 also point out that, if used naively for machine learning, Shapley value is not robust to replication (albeit in a different market setup than ours).
 We design a market that is robust-to-replication.
The key idea behind our approach is to allow a party to contribute data only if it also
submits a validation task
for which it requires a trained model.
Submitting a task, in turn, requires a participation fee.
Hence, for every replica, a party has to pay a participation fee that depends on its validation task.
As a result, the fee it pays for the improvement on its task balances the payoff it gets from the use of its training data.

Finally, we ensure that our cloud-based market is secure and
can be trusted by the parties.
That is, the market itself cannot resell, copy or release the data or the trained models.
Such a market can be instantiated using secure hardware, such as Intel SGX~\cite{sgx}, which
allows parties to submit their data encrypted
and ensure that it is only decrypted and processed in secure enclaves
and data is always encrypted in memory. It also provides attestation capabilities that
parties can use to verify the integrity of the market (e.g.,
that the data was used only to train specific models and was not copied in plaintext
outside of an enclave).
See~\cite{Ohrimenko2016} for details.

We summarize our contribution as follows:
\begin{enumerate}
  \item \emph{Marketplace Definition}: We introduce a collaborative marketplace to sell data and buy machine learning models for learning single and multiple validation tasks.
\ifFull
\else
  \item \emph{Payment Division}: We propose a novel and robust-to-replication payment division function.
\fi
  \item \emph{Customized model training}: We propose to release customized models to each party and not data in order to prevent malicious parties from reselling the data while
  incentivizing them to submit honest validation tasks (which in our market can be seen as an indirect bid on the data).
  \item  \emph{Evaluation}: We empirically evaluate the properties of our marketplace.
 Our experimental results confirm that our marketplace generates fair payoffs, customized models and is robust to replication. 
\end{enumerate}



\section{Background and Notation}

\ifFull
\else
\paragraph{Machine Learning Models and Evaluation.}
\fi
\label{sec:mlbg}

In supervised machine learning, we often consider training of models for regression or classification tasks, i.e.,
learning of a prediction function $\model \colon \Omega \rightarrow \mathbb{R}^k$ or $\model \colon \Omega \rightarrow [K]$, respectively, where $\Omega$ is some input space, $k \in \mathbb{N}$, and $[K]=\{1, \ldots, K\}$.
We introduce concepts for classification problems only but the same concepts apply to regression problems.
The model $\model$ is a part of some hypothesis space $\mathcal{H}$, e.g., the set of one layer neural-networks with a fixed number of hidden neurons and sigmoid activations.
These functions are learned from labeled training data $\trdata$, i.e., collections/sets of tuples of the form $(\omega, z)$, where $z \in [K]$ for classification.
Learning in that context commonly refers to minimizing a sample-wise loss function $l\colon \Omega \times [K] \rightarrow \mathbb{R}$, e.g., the cross-entropy loss, overloading notation: $\model(\trdata) = \arg \min_{\model' \in \mathcal{H}} \tfrac{1}{|\trdata|} \sum_{(\omega, z) \in \trdata} l(\model'(\omega),z)$.
The goal of the learning process is to identify functions $\model$ that perform well on unseen data, i.e., test data, and, for instance, achieve good classification accuracy. 
As a proxy, for estimating the performance on test data, we use validation data~$\valdata$.
The average classification performance on the validation data is $\gain(\valdata,\model(\trdata)) = \tfrac{1}{|\valdata|} \sum_{(\omega, z) \in \valdata} \mathbf{1}_{\model(\omega) = z}$, where $\mathbf{1}$ is the indicator function.
Clearly, the performance measure is application dependent and can, for instance, also be the RMSE, ranking accuracy, etc.
We refer to $\gain$ as (performance) gain function.

\ifFull
\else
{\bfseries Properties of $\gain$~.}
\fi
In the paper and for simplicity, we assume an idealized non-negative gain function $\gain$ and idealized dependencies of the model on the training data.
For training datasets $\trdata, \trdata'$ we assume that:
\begin{enumerate}[(i)]
  \item \emph{replicated data does not change performance}: \label{itm:assumption_duplicates} $\forall \trdata, \trdata'\colon \model(\trdata) = \model(\trdata \cat \trdata)$; \label{lbl:prop1}
  \item \emph{monotonicty}: $\gain(\valdata, \model(\trdata))) \leq \gain(\valdata, \model(\trdata \cat \trdata')))$;\label{lbl:prop2}
  \item \emph{supermodularity\footnote{In our theoretical analysis, we show that supermodularity in combination with our other assumptions is sufficient to prove that our marketplace instantiations
  in \S\ref{sec:instantiations} are robust
  to data replication as introduced in \S\ref{sec:properties}. However, it is not a necessary condition.}}: $\gain(\valdata, \model(\trdata \cup \{x\})) - \gain(\valdata, \model(\trdata)) \leq \gain(\valdata, \model(\trdata' \cup \{x\})) - \gain(\valdata, \model(\trdata'))$ for $\trdata \subseteq \trdata'$ and $x \not\in \trdata'$;\label{lbl:prop3}
  \item \emph{boundedness}: $\gain(\valdata, \model(\trdata)) \leq 1$.
\end{enumerate}
Here, (\ref{lbl:prop1}) captures that duplication of training data does not change the learned model.
This is, for instance, true for 1-NN classifiers;
(\ref{lbl:prop2}) characterizes that additional data either improves or does not change the performance;
(\ref{lbl:prop3}) captures the collaborative market setting where
each party provides complementary data that contributes towards a ML task
even in the presence of other parties' data.


\ifFull
\else
\paragraph{Shapley Values for Fair Payoffs.}
\fi
Consider a (machine learning) task which $M$ parties $\parties = \{1, \ldots, M \}$ aim to solve with a joint effort.
To quantify the value of the contribution of each party towards solving the task, we consider a characteristic function $v\colon 2^{\parties} \rightarrow \mathbb{R}$.
For every set $S \subseteq \parties$ of parties, $v(S)$ quantifies how well the parties in $S$ can solve the task, e.g., $v(S)$ could be the prediction accuracy of the best model the parties in $S$ can train by combining their training data, i.e., $v(S)=\gain(\valdata, \cup_{i \in S} \trdata_i),$ where $\trdata_i$ is the $i$th party's training data.
When the characteristic function
depends on the maximum number of participants,
we overload the notation and use $\vmult(S, \parties)$
to denote the value parties in the set $S\subseteq\parties$ bring to the market
with participants in $\parties$.

If parties are to be compensated for helping to solve a machine learning task, it is natural to ask what a \emph{fair} payoff for each parties' effort is. 
To this end, we consider Shapley values~\cite{shapley1953value}, i.e., unique payoffs, studied in game theory for collaborative games, that satisfy certain natural fairness properties discussed later. 
The Shapley value for characteristic function $v$ and party $i \in \parties$ is

\ifShort
\vspace{-5pt}
\fi
\begin{linenomath*}
\begin{align}
      \shap(v,i) \!= \!\!\!\!\!\!\! \sum_{S \subseteq \parties \setminus \{ i \}} \!\!\!\!\!\!\!\frac{|S|! (M-|S|-1)!}{M!}  \big( v(S\!\cup\!\{i\}) \!-\! v(S) \big),
      \label{eq:shapley}
\end{align}
\end{linenomath*}
i.e., $\shap(v,i)$ quantifies the average marginal contribution of party $i$ wrt all possible subsets of parties.
If $v(\parties) \neq 1$, Shapley values
can be normalized to sum to $1$ by scaling them by $1/v(\parties)$.
When clear from the context which characteristic function is used, we use $\shap(i)$.
\ifFull

\fi
\ifShort
For a fixed characteristic function $v$, Shapley values $\shap(i)$ are the unique payoffs satisfying properties of \emph{efficiency, symmetry, linearity and null player} (explained in \S\ref{sec:properties}). 
These properties ensure that parties are paid equally for equal contributions and all gains are distributed among the parties.
\fi

\ifFull
\else
We will use Shapley values with different characteristic functions in our marketplace: 
$u$ for deciding which data to use for training customized models (\S\ref{sec:modeltrain}) vs.~$\vsing$ and $\vmult$ for computing the payoffs for single validation tasks (\S\ref{sec:singletask}) and multiple validation tasks (\S\ref{sec:multitask}) respectively.
This is necessary as Shapley value itself is not robust to replication.
\fi



\section{A Collaborative Marketplace}
\label{sec:collabmarket}

 In our marketplace, we consider $M$ parties $P_1, \ldots, P_M$ which aim to collaborate
 towards training machine learning models for their tasks.
 Each party simultaneously takes the role of a seller and a buyer.
 The $i$th party has training data $\trdata_i$ and validation data $\valdata_i$.
 The performance of a machine learning model $\model(\trdata)$ trained on some training data $\trdata$ is evaluated
 using a performance/gain function $\mathcal{G}(\valdata_i, \model(\trdata)) \in [0,1]$.
 
 The goal of the market is to provide party $i$ with a customized model
 trained on the subset of (or potentially all) datasets of other parties' that best fits
 its task (based on $\valdata_i$). At the same time the market
 uses $\trdata_i$ to train models
 of parties where this dataset fits the corresponding task (based
 on validation data of other parties).
  
  \begin{definition}[Marketplace] 
  A \emph{marketplace} is a tuple $(\mathcal{P}, \textrm{PD})$, where $\mathcal{P} = (P_1, \ldots, P_M)$ is the list of parties engaging with the market place and $\textrm{PD}$ is the payment division function.
\end{definition}
  
  We consider the following interaction with the data marketplace:
  \begin{enumerate}
    \item  Parties $P_1, \ldots, P_M$ arrive.
    \item The marketplace collects all training data sets $\trdata_1, \ldots, \trdata_M$ and validation tasks $\valdata_1, \ldots, \valdata_M$.
    \item Every party pays the market a participation fee $\fee_i$ that is determined proportional to 
    unit increase in personal performance gain
    $\fee_i = 1 - \gain(\valdata_i, \model(\trdata_i))$, i.e., valuation for increasing performance to $100\%$ on their validation data.
    The market can also set a fixed value for a unit increase $c> 0$ in performance,
    such that $\fee_i = c \cdot (1 - \gain(\valdata_i, \model(\trdata_i)))$.
    \item The marketplace trains a machine learning model $\model^i$ for every party $i$ where $\model^i = \model(\valdata_i, \cat_{j \in S_i} \trdata_j)$ and $S_i \subseteq \parties$.
    \item The model $\model^i$ is shared with party $i$.
           Let $\price_i = \gain(\valdata_i, \model^i) - \gain(\valdata_i, \model(\trdata_i))$.
      Party $i$ receives payoff $t_i$ which depends on the increase in performance on its validation data $\valdata_i$ (i.e., they receive $\fee_i - \price_i$)
     and how much its data helps in improving performance of models for other validation tasks, $b_i$.
     Hence, in total party $i$ gains (or loses) $(\fee_i - \price_i) + b_i - \fee_i = b_i - \price_i$.
  \end{enumerate}


\subsection{Desired Properties} \label{sec:properties} In the following we enumerate desired properties for a machine learning marketplace such that participating parties receive \emph{fair} payoffs for their engagement and benefit from participation.

\paragraph{Revenue division and payment.}
 Our list of properties is inspired by the  
``standard axioms of fairness'' since they are the de facto method to assess
the marginal value of goods (i.e., features in our setting) in a cooperative game (i.e., prediction task in our setting).
They include:
\ifFull
\else
\emph{Balance:} $\sum b_i = \sum a_i$.
\emph{Symmetry:} if two  parties  $i$ and $j$ enter the market with same training and validation sets then the their payoff, $t_i = t_j$.
\emph{Zero element buyer:} if there is a party whose performance does not increase, it should at least get its participation fee back, $t_i = A_i$.
\fi

\paragraph{Incentives.} 
Party $i$ decides on whether to enter the collaborative market or not
and what $\valdata_i$ and $\trdata_i$ to contribute.
As a result, the marketplace should incentivize the parties to join the market
with good training data and honest validation data:

\emph{Joining the market:} The market should incentivize new parties to join.
Our setting incentivizes this as follows.
A new party brings a new task to the market, hence, existing parties'
data may be useful for this task, increasing their payoff $b_i$.
At the same time, the new party is also bringing new training data which can increase performance of tasks submitted in the existing market, resulting in an increased payoff.
Hence a party joining the market can benefit from increased utility.

\emph{Validation data:} A dishonest party can try to manipulate training and validation data (including their relationship) in order to gain
more than it would with its true training and validation datasets.
For example, $\valdata_i$ can be seen as an implicit bid that party $i$
places on the model $\model^i$ it will obtain.
There can be a case that
\ifFull
\else
$\gain(\valdata', \model^i) - \gain(\valdata', \model(\trdata_i)) > \gain(\valdata_i, \model^i) - \gain(\valdata_i, \model(\trdata_i))$.
\fi

If it is the case, $\model^i$ has higher utility than what is determined by
$\valdata_i$.
To this end, the marketplace needs to ensure that
the model $\model^i$ that is returned to the party does not allow for existence
of $\valdata'$ in the current marketplace, incentivizing the party to provide the best $\valdata_i$
to get the best utility model.
We enforce this incentive by training models that are customized for a specific task,
i.e., maximizing their accuracy towards a specific task while minimizing their accuracy on all other tasks in the market.

\emph{Training data:} The market needs to ensure that the payment $b_i$
that party $i$ receives for the use of its data to train other models
incentivizes it to provide its best $\trdata_i$.
We note that our market does not have an explicit way for parties to bid or price training data.

\paragraph{Robustness to replication.}
Parties may not behave honestly and
may replicate their data and create new replica parties to join the market
on their behalf.
The market should be robust to replication.
That is, a party that replicates its training data
should not earn more
than it would in the original market.
This is a crucial property for any data market since
data as compared to physical goods is easily replicable.
This problem was already highlighted for a different marketplace setup in~\cite{agarwal2019marketplace}.

\paragraph{Privacy and Integrity.}
Since data can be easily leaked and manipulated compared to physical goods,
the market should provide assurance to its participants
about the correct handling of their data and model training.
At the very least it should ensure that information
about the training data $\trdata_i$ of party $i$ is revealed only to other parties
through models for validation tasks where $i$th training data is useful.
The information about $\valdata_i$ has to be kept secret
as well.
As we mention in the introduction, these properties
can be provided if the market infrastructure is instantiated using Trusted Execution Environments
(or enclaves):
parties submit their data encrypted and allow only attested (verified) code
to decrypt, compute on it, and finally encrypt the
models under the encryption keys of the parties
who can access these models (i.e., the party who specified the
validation data for this model).



\ifFull
\else
\section{Market Instantiations}  
We provide market instantiations for both a single validation task among all the parties and multiple validation tasks.
\label{sec:instantiations}

\subsection{Single Validation Task}  
\label{sec:singletask}
\fi

Let us describe the marketplace for a single validation task~$\valdata$ that all $M$ parties agree on.
We slightly abuse the notation by letting $\model_S =  \model(\cat_{k \in S}\trdata_k)$
be the model trained on data from parties in $S$ for task $\valdata$.
The parties agree on the same marginal payment per increase of
the model performance they gain (for example, per percentage increase):
if $\gain(\valdata,\model_{\parties}) = 1$ then
party $i$ pays the amount proportional to
$\fee_i = 1- \gain(\valdata,\model_{i})$.
Hence, the largest amount that can be distributed among
market participants is $\sum_{i\in\parties} \fee_i$.
Recall that, when entering the market the party submits $\trdata_i$
and fee $\fee_i$.
After the market completes,
$i$ obtains $\model_{\parties}$
and payout $b_i \ge 0$.

\ifFull
\else
In \S\ref{sec:modeltrain}
we explain how to train $\model_{\parties}$ from
parties' datasets
and here describe
how we instantiate the market,
compute the payoffs using Shapley and
show that the market is robust to replication.
\fi

\ifFull
\else
\paragraph{Characteristic Function.}
\fi
\label{sec:charfuncsingle}
Our characteristic function captures 
the value of data to parties in a set $S$ for the task $\valdata$ as the value of the model trained on the data as well as the
value the model brings to every party.
As a result the characteristic function for this market is defined as
\begin{eqnarray}
\vsing(S) =  \underbrace{\gain\big(\valdata; \model_{S})}_{\text{value of the model}} + \sum_{j \in S} \big[ \underbrace{\gain\big(\valdata; \model_S) - \gain\big(\valdata; \model_j)}_{\text{model value for party $j$}} \big]
\label{eq:charS}
\end{eqnarray}
This function can be seen as the value
of the model trained on the datasets of all parties in $S$ plus the marginal gains
for each party.
Note that for a single party the value of the data is expressed as the value
of the model trained on its own training dataset.

\newcommand{\totalval}{\mathbf{a}}

\ifFull
\else
\paragraph{Payment Division.}
\fi
\label{sec:singlepaydiv}
The total amount, $\totalval$, that is distributed among the participants depends
on the individual gains obtained from the final model.
Let $\price_i = \gain(\valdata, \model_{\parties}) - \gain(\valdata, \model_{i})$.
Then,
\ifFull
\else
$\totalval = \sum_{i\in\parties} \price_i$.
\fi
We use (normalized) Shapley values for characteristic function~$\vsing$ to determine the distribution
of~$\totalval$ to each party~$i$:
$b_i = \totalval \times \shap(\vsing,i)$.
To simplify the notation we use $\shap(i)$ to denote $\shap(\vsing,i)$.
In total, party~$i$ obtains $t_i = (\fee_i - \price_i) + b_i$
where $(\fee_i - \price_i)$ is the return of the original investment
if the final model performance is not~$1$.
Hence, party $i$ gains/loses the following amount 
by participating in the market:
\ifFull
\else
$(\fee_i - \price_i) + b_i - \fee_i = \shap(i)\totalval - \price_i$.
\fi
Note that $i$ gets $\shap(i)$ portion from each $\price_j$ and it pays $1-\shap(i)$ of~$\price_i$ to the market.

\ifFull
\else
\paragraph{Properties.}
\fi

Our single-task market has the following properties:
\emph{Balance:} $\sum_i (b_i - \price_i) = \sum_i \totalval \times \shap(i) - \sum_i \price_i = \totalval \times 1 - \totalval  = 0$.
\ifFull

\fi
\emph{Symmetry:} This follows from using Shapley value to calculate payout $b_i$ and $\price_i$
would be the same for the parties with same data.
\ifFull

\fi
\emph{Zero element buyer:}
If $i$ does not benefit from other parties, that is  $\gain(\valdata,\model_i) = \gain(\valdata,\model_\parties)$,
then $\fee_i$ is returned to $i$ since $\price_i = 0$. 
\ifFull

\fi

\ifFull
\else
\paragraph{Robustness to replication.}
\fi
%
\ifFull
\else

Recall that a market is robust if
a party does not gain a higher payoff in the market with the replicas
compared to its gain in the original market.

Let us consider the total payoff of $i$ when it replicates itself once.
Let $\totalval$
and $\totalval^R$ 
be the values of the original and replicated markets.
By definition, $\totalval^R = \totalval+\price_i$.
Let $t$ denote $i$'s payoff in the original market
and $t^R$ be its total payoff when it replicates itself.
Similarly, let $\shap(\cdot)$ and $\shapR(\cdot)$ denote Shapley values before
and after replication.
\ifFull
\begin{eqnarray*}
t^R = 2 \times \shapR(i) (\totalval^R - \price_i) - 2\times (1-\shapR(i)) \price_i = \\ 2 \times ( \shapR(i) \totalval^R - \phi_i(v)  \price_i - \price_i +\shapR(i) \price_i ) =\\ 2 \times ( \shapR(i) \totalval^R - \price_i ) 
\end{eqnarray*}
\else
$t^R = 2 \times \shapR(i) (\totalval^R - \price_i) - 2\times (1-\shapR(i)) \price_i = 2 \times ( \shapR(i) \totalval^R - \price_i )$.
\fi
In the market where each party is equivalent, that is $\forall k, \price_k = \price$,
replication does not help.
\ifFull
Using the symmetry property of Shapley value:
\begin{eqnarray*}
 t^R = 2 \times ( \shapR(i) \totalval^R - \price ) =  2 \times (\frac{1}{M+1}\sum^{M+1} \price - \price ) =  2 \times (\price - \price) = 0 
\end{eqnarray*}
\else
Using the symmetry property of Shapley value we can show that $t^R = 0$.

\fi
Let us consider the case when the contributions of the parties
are different.
Recall that:
$t=\shap(i)\totalval - \price_i$
and 
$t^R = 2 \times ( \shapR(i)(\totalval + \price_i) - \price_i )$.
Replication is useful to party $i$ only if $t < t^R$.
If this is the case, then:
\ifFull
\begin{eqnarray*}
t &<& t^R\nonumber\\
\shap(i)\totalval - \price_i &<& 2(\shapR(i)(\totalval+\price_i) - \price_i)\nonumber\\
\shap(i)\totalval + \price_i&<& 2\shapR(i)(\totalval+\price_i) \nonumber\\
\end{eqnarray*}
\else
$\shap(i)\totalval + \price_i< 2\shapR(i)(\totalval+\price_i)$ should hold.
\fi
\begin{condition}
The market in \S\ref{sec:instantiations} is robust to replication if
$\shapR(i) \le \frac{\shap(i)\totalval + \price_i}{2(\totalval+\price_i)}
$.
\label{cond:cmpshap}
\end{condition}

We show that our single-task market instantiation is robust to replication
as long as the gain function $\gain$ has the properties outlined in \S\ref{sec:mlbg}.
Our proof proceeds as follows.
We first derive the payoff that $i$ needs
to make in the replicated market in order
to break even.
We then determine 
the marginal contribution
of party $i$ for new coalitions
that are created due to the addition of its replica~$i'$ (see Lemma~\ref{lemma:margcontrib}
in Appendix).
We then devise
the relationship between
the Shapley values of party~$i$ before
and after replication.
Finally we show that the Shapley value of $i$ in the new replicated market
satisfies the condition of robustness to replication, yielding the following theorem (proof details are provided in the appendix):
\begin{restatable*}{theorem}{thmrobustsingle}
\label{thm:robustsingle}
The single-task market instantiation in \S\ref{sec:singletask} is robust to replication
as per Condition~\ref{cond:cmpshap}.
\end{restatable*}
\fi

\ifFull
\else
\subsection{Multiple Validation tasks}  
\label{sec:multitask}
\fi

The market setup is similar to the single task market
where there is a task from each party, $\valdata_i$.
Similarly
when entering the market the party submits $\trdata_i$
and fee $\fee_i$.
After the market completes,
$i$ obtains $\model^i_{\parties}$
and payout $b_i \ge 0$.
In \S\ref{sec:modeltrain}
we explain how to train an individual $\model^i_{\parties}$ from
parties' datasets
and use the rest of the section to describe
how we instantiate the market
and compute the payoffs.
%

\vspace{-10pt}
\paragraph{Characteristic Function.}  
Let $\partiesD \subseteq \parties$ be a set of indices of the distinct validation tasks,
i.e., for every $i \in \parties$, there is $j \in \partiesD$ s.t.~$\valdata_i = \valdata_j$.
Then the characteristic function for the multi-task market is defined as:
\begin{eqnarray}
\vmult(S,\parties) &=&  \underbrace{\sum_{i\in \partiesD} \gain\big(\valdata_i; \model^i_{S})}_{\text{value of models for distinct tasks}} \nonumber \\
&+& \sum_{i \in S} \big[\underbrace{\gain\big(\valdata_i; \model^i_{S}) - \gain\big(\valdata_i; \model^i_{i})}_{\text{value for party $i$}} \big] \label{eq:charM}
\end{eqnarray}
When clear from the context, we omit the second argument of $\vmult(\cdot, \cdot)$.
We observe that this is a natural extension from a single validation task:
the goal of the market is to achieve better performance on all (distinct) tasks of the market as well as individually for
each party.
As we will see it also helps with lowering the impact of replicated parties on
the value of the market: an addition of a replicated party will affect only the second part of the value function
while a party with new data has a chance to contribute to both parts of the characteristic function.

\ifFull
\else
\paragraph{Payment division and properties.}
\fi
The payment division is identical to the one for a single task market in \S\ref{sec:singlepaydiv}
except that $w$ is used as the characteristic function when computing the Shapley value.

\ifShort
Since the payment function is the same as in the single-task market,
the condition on robustness is also the same
and marginal contribution of $i$ to coalitions that
are created due to its replication is limited in the same
(see Lemma~\ref{lemma:margincontribM} in Appendix).
\fi

Since multi-task market payout is defined
in the same manner as the one for single task,
the following properties also hold: balance, zero element
and symmetry.
The main difference is the computation
of the characteristic function~$w$.
The multiple-task market is robust to replication as we
prove in the appendix.
%

\begin{restatable*}{theorem}{thmrobustmult}
\label{thm:robustmulti}
The multiple-task market instantiation in \S\ref{sec:multitask} is robust to replication
as per Condition~\ref{cond:cmpshap}.
\end{restatable*}



\begin{figure*}[h]
    \begin{subfigure}[b]{.49\textwidth}
        \centering
        \includegraphics[width=0.49\textwidth]{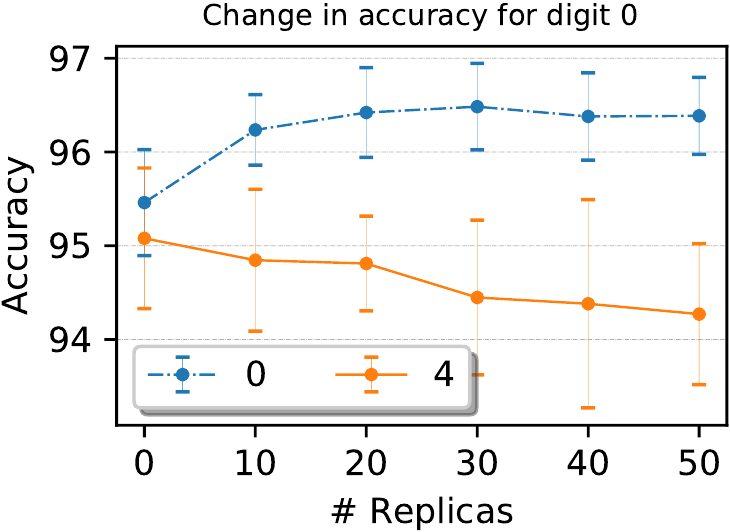}
        \includegraphics[width=0.49\textwidth]{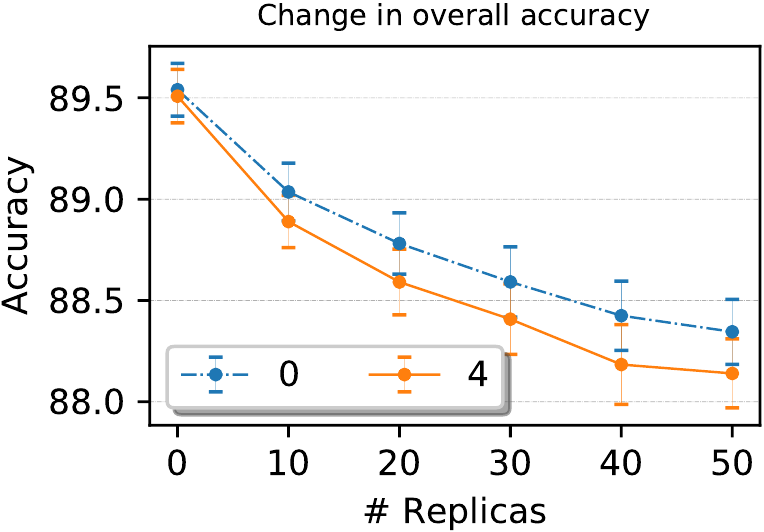}
        \label{fig:mnist_lr_0}
        \caption{Logistic regression}
    \end{subfigure}%
    \hfill
    \begin{subfigure}[b]{.49\textwidth}
        \centering
        \includegraphics[width=0.49\textwidth]{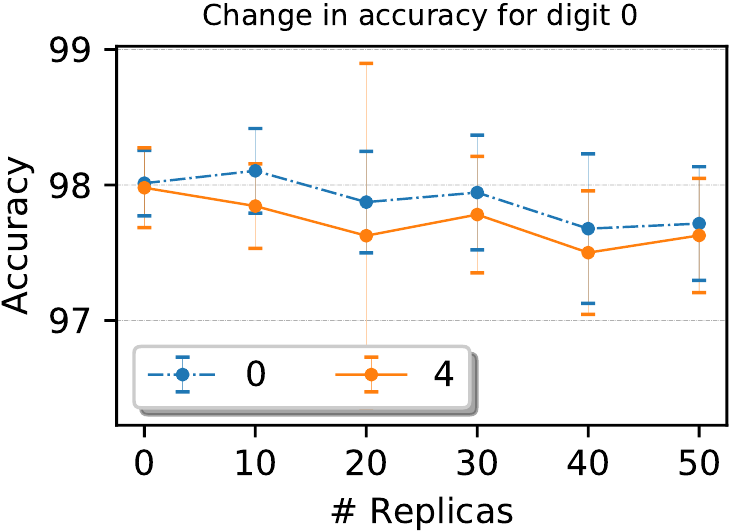}
        \includegraphics[width=0.49\textwidth]{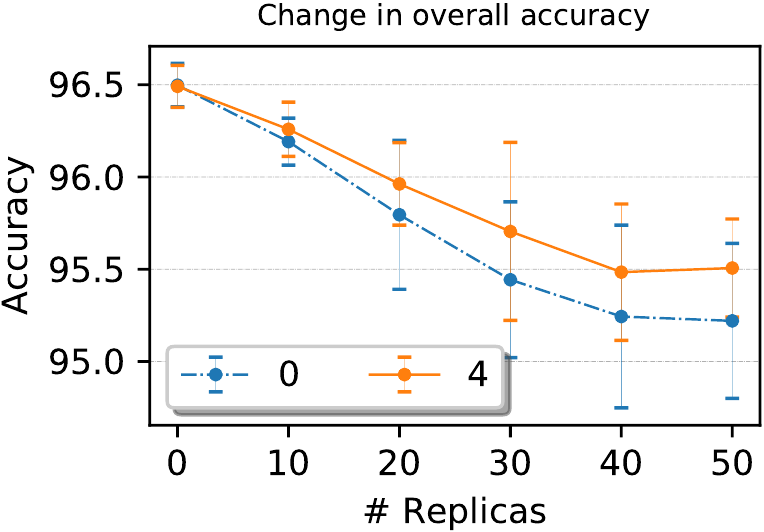}
        \label{fig:mnist_nn_0}
        \caption{DNN}
    \end{subfigure}%
\hfill
\caption{Accuracy for MNIST dataset trained using logistic regression (two left plots) and a DNN (two right plots)
with replication of digit~$0$. Replication of other digits showed similar results.}
\label{fig:Mnist}
\end{figure*}

\section{Training over Multi-Party Data}
\label{sec:modeltrain}
In this section, we describe
(1) how to align the training data provided by all parties for training models
for different tasks, i.e., select the relevant training data;
and (2) how to train a model for party $i$
such that it performs well only on its corresponding task and
not other tasks in the market.

\ifFull
\else
\paragraph{Training Data Selection.}
\fi
\label{sec:data-selection}

With multi-party data we need to select the right training data for training each parties' model.
To this end, we also make use of Shapley values. 
Note that this step happens prior to the computation of the payoffs which again involves the computation of Shapley values but for a different characteristic function.
For training data selection for party $i$ we follow two steps:
\ifFull
\else
(1) Compute the Shapley values $\shap(u_i,i)$ for $u_i(S) = \mathcal{G}(\valdata_i, \model(\cup_{j \in S} \trdata_j))$.
(2) Determine the relevant data for the task of party $i$ as $\mathcal{D}_i = \{j \in \parties \mid \shap(u_i,j) \geq \tau\}$, where $\tau$ is a threshold that can be used to control the amount of used training data.
\fi
Note that the above steps can be trivially applied on a sample-level instead of a dataset-level.
Computing sample-level Shapley value would typically result in better model performance at increased computational cost for deciding on the training data to use.
In the case of a single common validation set, $\mathcal{D}_i = \mathcal{D}_j$, $\forall i \neq j$.

The characteristic functions used for computing payoffs in \S\ref{sec:multitask} are computed using only relevant training data. 
In particular, when computing Shapley values for determining each parties' payoff, we use the shorthand $\model_{S}^i$ for $\model_{S \cap \mathcal{D}_i}$.
\ifFull
\else
Superscript is omitted for single validation task.
\fi

\ifFull
\else
\paragraph{Customized Model Training.}
\fi
\label{sec:adversarial-training}

To minimize usefulness of the model provided to party $i$ after interacting with the market for other parties' validation tasks, we adopt the following strategy for training the model for party $i$:
Party $i$ receives a model trained on the data of parties $\mathcal{D}_i$.
Additionally the model that party $i$ receives is optimized for maximizing the loss wrt all $\valdata_k$ for $k \neq i$, i.e.
\begin{align}
  \min_{\theta} \qquad &\sum_{j \in \parties \setminus \{i\}} \mathcal{G}(\valdata_i, \model(\cat_{k \in \parties} \trdata_k)) \\
       \textnormal{s.t.} \qquad & \mathcal{G}(\valdata_i, \model(\cat_{k \in \parties} \trdata_k)) \geq \gain^* - \epsilon, \nonumber
\end{align}
where $\gain^*= \mathcal{G}(\valdata_i, \model^i_\parties)$ is the best performance achievable on the validation data of party $i$ by standard model training, $\epsilon \geq 0$ is some constant, and $\theta$ are parameters of the models, e.g., weights of a neural network.
To avoid overfitting to the validation data by using it during model training (we still need to resort to the validation data for pricing) and have a straightforwardly approachable optimization problem, we consider the proxy
$\max_\theta \  \gain(\cup_{k \in \mathcal{D}_i} \trdata_k, \model^i_{\parties}) - \lambda \sum_{j \neq i} \gain(\valdata_j, \model^i_{\parties}),$ 
where $\lambda > 0$ is a hyperparameter.
This approach is evaluated in our experiments.



\section{Experiments}
\label{sec:experiments}

We implemented a prototype of our marketplace
and evaluated it on the task of classifying handwritten digits from the MNIST dataset~\cite{mnist}.
Here we validate
assumptions made in the paper
and evaluate the effectiveness of our marketplace setup.
The goals of our evaluation are:
\begin{inparaenum}[(i)]
	\item Understand the effect of replicated training data on the accuracy on validation data;
    \item Measure the (importance) value of training data towards a single validation task;
       \item Calculate the payoffs with our payment division function for training data with and w/o replication
    and compare them to those of the naive payment division; 
     \item Evaluate the effectiveness of an output model customized for a given validation task.

\end{inparaenum}

\subsection{Effect of Data Replication on ML models}
\label{sec:exp-replication}
We aim to understand how data replication
affects model accuracy and validate the assumption we made
in \S\ref{sec:mlbg}.
To this end, we consider a collaboration among $10$ parties
for the validation task of $10$K images with all $10$ digits. 
We measure the accuracy on this validation task on models that are trained  using the data from all parties.
We train a logistic regression model and a single hidden layer (500 neurons) DNN model for 100 epochs with a learning rate of $10^{-2}$, Adam optimizer and a value of $10^{-4}$ for $\ell$2 regularization.

In the initial setup, each party contributes $1000$ images corresponding to a single digit,
which are randomly sampled from the training data for this digit.
To understand the effect of replication, we compare
the accuracy of the model trained in the initial setup with the following replication configurations.
The party with digit $0$ or $4$, respectively, replicates and creates new parties with the same $1000$ images as their training data.
We vary the number of replicas from $0$ to $50$.
Hence, there will be $50{,}000$ samples corresponding to the digit $0$ in the combined training dataset for $50$ replicas.
All observations are similar for the replication of other digits.
In addition to the overall accuracy of the validation task, we looked at the accuracy of the digits $0$ and $4$ present in the validation dataset.

\paragraph{Results.}
\label{sec:exp-replication-app}

Figure~\ref{fig:Mnist} shows results for a logistic regression and a DNN model averaged over 50 runs (random draws of training and validation data, and of the initial network weights).
The dotted and the solid lines denote the accuracy over the number of replicas.
For logistic regression, the accuracy of classifying digit $0$ increases initially with replication of digit $0$ and then remains constant.
This confirms that replication does not necessarily benefit the accuracy of the replicated digit on a model that is trained on sufficient samples.
On the other hand, the accuracy of classifying the digit~$0$ decreases slightly with replication of~$4$.
This result shows that replication may hurt the accuracy of other validation tasks in the market.
Finally, we observe that the accuracy of the overall validation task of $10$K images reduces only slightly with the replication of both digits.
The results for DNNs show similar trends.
Thus, in summary, the assumptions needed for our theoretical results hold approximately in practice.

\begin{figure*}[t]
  \centering
  \begin{subfigure}[b]{.33\textwidth}
    \includegraphics[width=1.\linewidth]{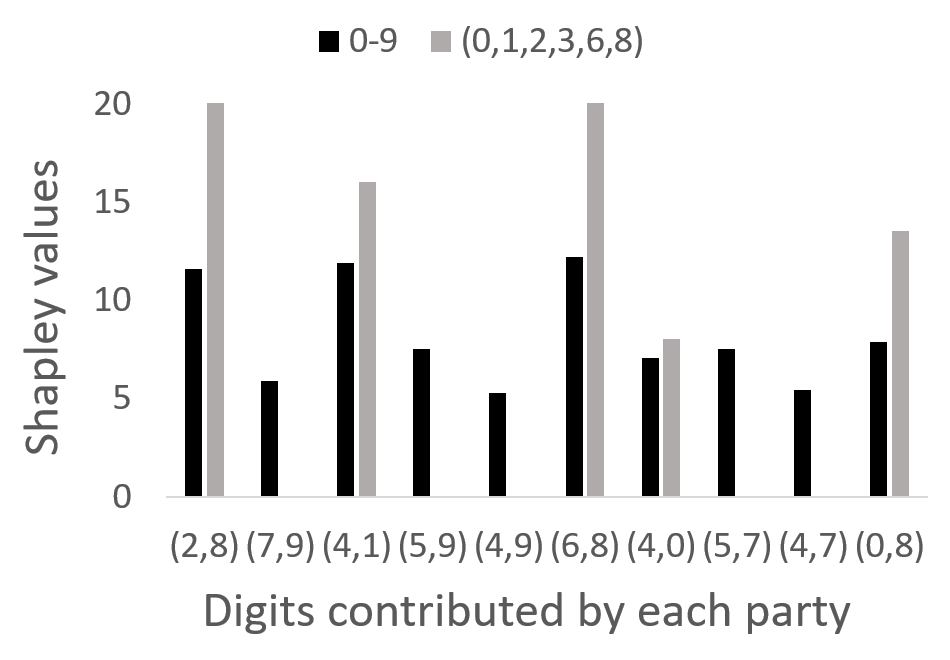}
    \caption{}
    \label{fig:mnist_shapley}
  \end{subfigure}
  \begin{subfigure}[b]{.33\textwidth}
	  \includegraphics[width=1.\linewidth]{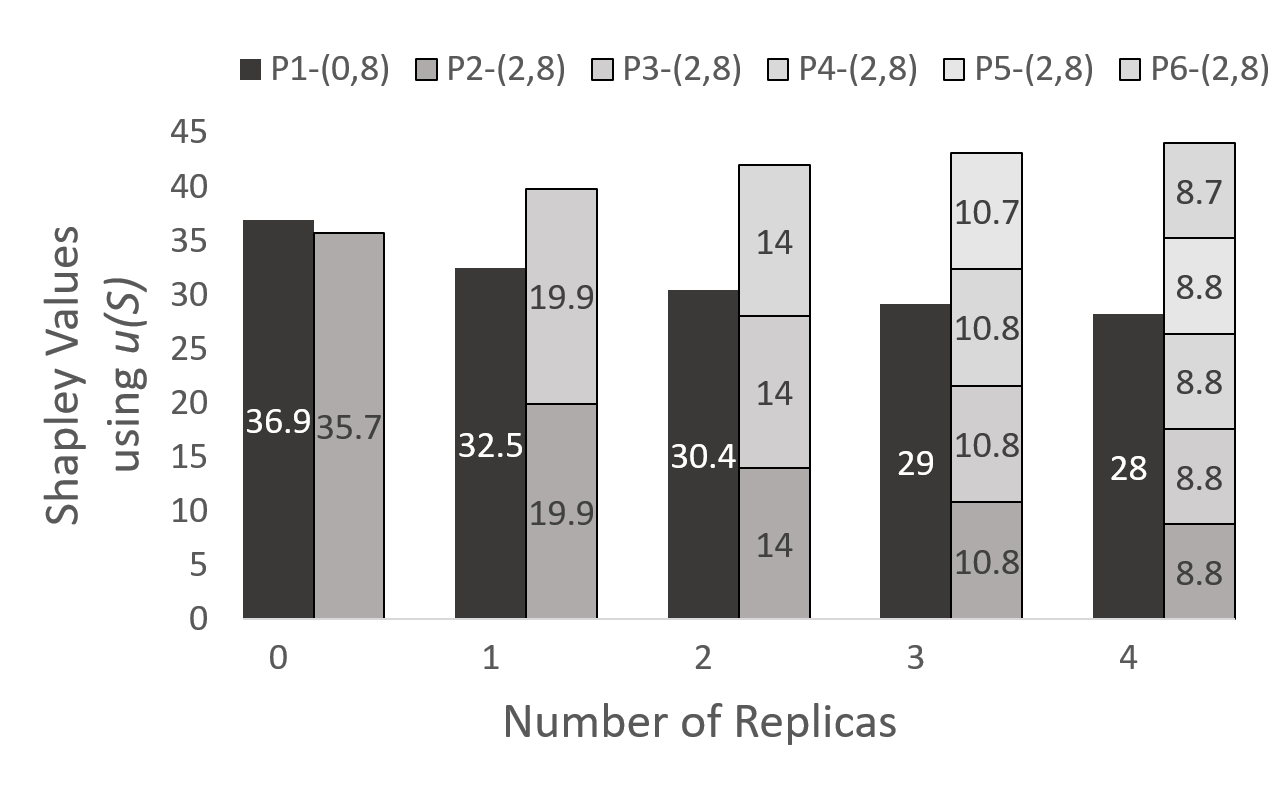}
    \caption{}
    \label{fig:mnist_shapley_replication}
  \end{subfigure}
\hfill
  \begin{subfigure}[b]{.31\textwidth}
    \includegraphics[width=1.\linewidth]{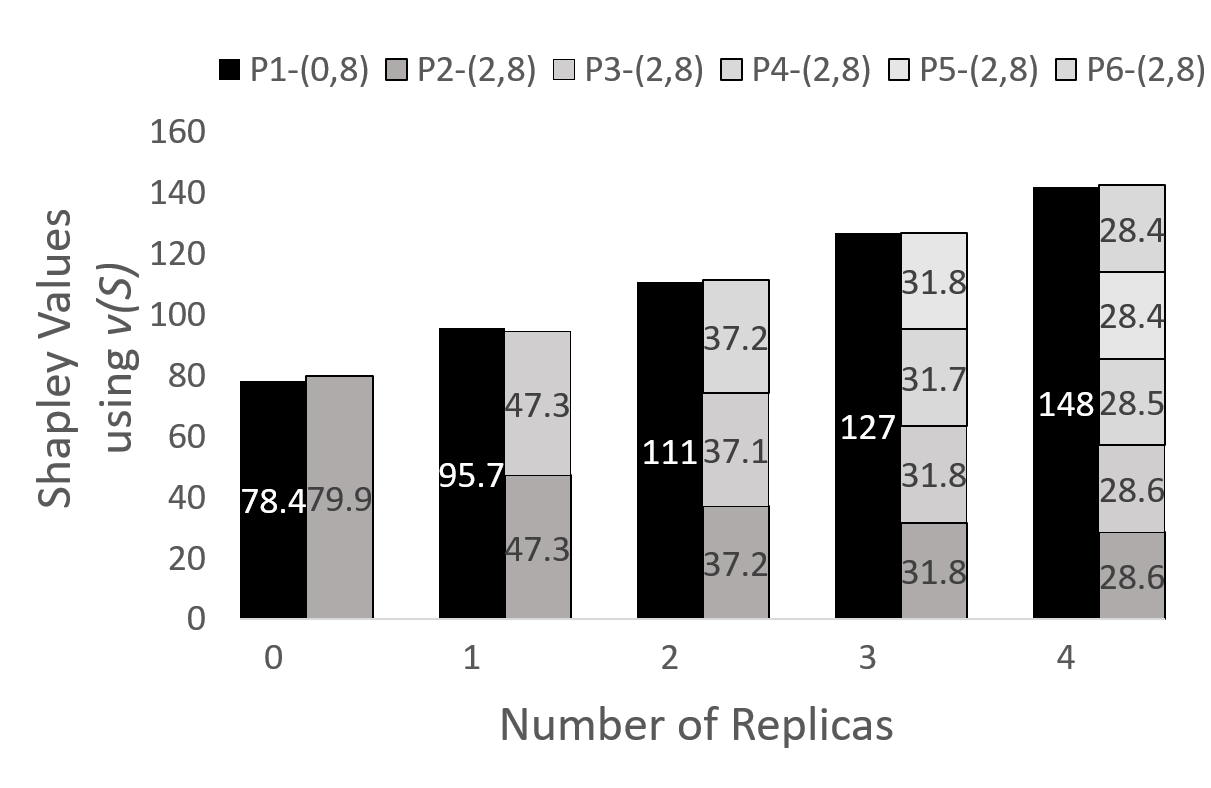}
    \caption{}
    \label{fig:payoffs}
  \end{subfigure}
  \vspace{-2mm}
  \caption{(\subref{fig:mnist_shapley}) Shapley values for parties contributing training data for validation tasks with all ten digits and only digits $(0,1,2,3,6,8)$.
           Figures (\subref{fig:mnist_shapley_replication}) and (\subref{fig:payoffs}): Shapley values 
            on validation dataset of digits (0,2,8)  with replicating parties for characteristic function $u$ (accuracy) in (\subref{fig:mnist_shapley_replication}) and our proposed characteristic function $v$ (\S\ref{sec:charfuncsingle}) in
            (\subref{fig:payoffs}). }
  \label{fig:shapley_values}
\end{figure*}

\subsection{Customized Model Training}
\label{sec:exp-trdata-selection}

An important component of our marketplace is customized model training, cf.\ \S\ref{sec:modeltrain}, which consists of two steps: data selection and model training.
We empirically investigate both steps here.

\textbf{Training data selection.}
To understand the usefulness of training data with respect to a validation task, we consider a setting where each party holds $2000$ training samples for two randomly selected labels.
As validation tasks, we consider classification of all ten digits and classifying digits $(0,1,2,3,6,8)$. 
Figure~\ref{fig:mnist_shapley} shows the Shapley values using characteristic function $u$ for both validation tasks.
For all the digits in the validation data (label \emph{0-9} in the figure), we observe that the parties with unique labels such as 1, 2 and 6 have higher Shapley values indicating higher utility
of their training data. For the validation task on digits $(0,1,2,3,6,8)$ we observe that the parties which do not contribute any of these digits have zero Shapley values.
This confirms that Shapley values are well suited for selecting training data for task-specific model training.

\textbf{Model training.}
To evaluate the customized model training we considered a setting with ten parties in which each party holds data for two different digits s.t.\ party $j$ has data for labels $(j,(j+1) \mod 10)$ for $j \in \{0, 1, \ldots, 9\}$.
For party $j$ the training data consists of all training samples from the MNIST dataset with either of its two labels.
The same holds for its validation data with respect to the validation data from MNIST.
In Figure~\ref{fig:model-usefulness} we see results for applying the customized model training approach
when training a logistic regression model for the validation task of each party.
We observe that although the combined training points consist of all the digits in MNIST, the model trained with our approach is useful only for the specific validation task and has mediocre performance on (other) validation tasks with overlapping labels. 
That is, it is not useful to classify digits from the remaining labels outside the specific validation task.


\subsection{Payoffs in Single Task Marketplace}
\label{sec:exp-payoffs}
We now empirically validate that our proposed characteristic function $v$ in \S\ref{sec:charfuncsingle} for computing payoffs is indeed robust to replication.
We also compare it to the case where Shapley values are computed for characteristic function $u$ (model accuracy), cf.~\S\ref{sec:modeltrain}, which is not robust to replication.
Here, we consider two parties: Honest party $P_1$ with digits (0,8), and replicating party $P_2$ with digits (2,8).
$P_2$ creates replicas $P_3$ to $P_6$ with the same dataset.
The single validation task is classification of digits 0,2,8.
The parties combine their data and train a logistic regression model.

Figure~\ref{fig:shapley_values} shows Shapley values of our results as number of replicas increases.
We observe that the combined Shapley value of the replicating parties increases with the number of replicas while that of the honest party $P_1$ decreases for~$u$ (Figure~\ref{fig:mnist_shapley_replication}) when using
Shapley values with accuracy as the characteristic function. 
In contrast, Figure~\ref{fig:payoffs} shows that when Shapley values are instantiated with
our new characteristic function, Shapley values of $P_1$ and the combined values for the replicating parties are equal.
This verifies that replication does not benefit the replicating party when using
our single-task collaborative market instantiation.

\begin{figure}[t]
	\centering
    \includegraphics[width=1.\linewidth]{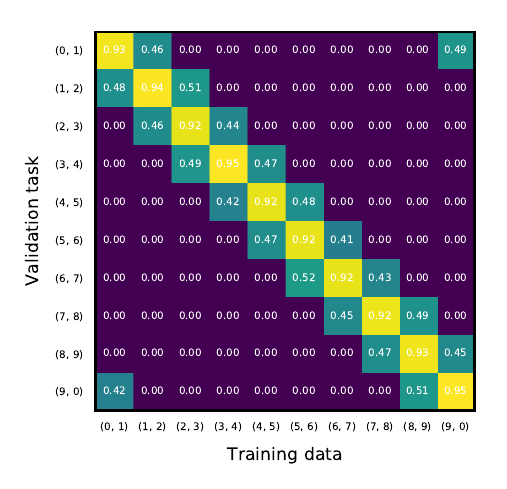}
    \caption{Classification accuracy of the customized models for different validation tasks.}
    \label{fig:model-usefulness}
\end{figure}



\section{Related Work}
Multiple lines of work are related to our proposed collaborative machine learning markets framework.
We organized it into work on using \emph{Shapley values for valuating data} for model training and \emph{machine learning marketplaces}. 
Also general work on game theory and mechanism design is related to our work but not discussed here in detail, cf.~\cite{nisan2007algorithmic} for an overview.

{\bfseries Shapley values for valuating data.}
Prior work has proposed the use of Shapley values for characteristic function $u$ to valuate data in machine learning settings.
Most of these works focus on the problem of efficiently approximating Shapley values.
Closest to ours,~\citeauthor{agarwal2019marketplace}~\shortcite{agarwal2019marketplace} uses Shapley values to compute payoffs in non-collaborative marketplace. 
Before that, Datta et al.~\shortcite{datta2016algorithmic} proposed the use of Shapley values to quantify feature importance for classification problems, and proposed sampling based approximations of Shapley values.
Ghorbani et al.~\shortcite{ghorbani2019} show that Shapley values are good metric to quantify the usefulness of data for machine learning tasks.
They show that valuating data for model training based on Shapley value is better than using leave-one-out cross-validation, and  propose sampling and gradient based approaches for approximating Shapley values.
Also Jia et al.~\shortcite{jia2019towards} propose several sampling based approaches and the use of influence functions~\cite{koh2017understanding} to reduce the computational cost of computing Shapley values.

{\bfseries Machine learning marketplaces.}
The closest to our paper is work by Agarwal et al.~\shortcite{agarwal2019marketplace} who study an algorithmic approach for data marketplaces.
They propose a general marketplace setup which makes use of Shapley values for computing payoffs.
The marketplace sets prices for data in an online fashion and matches sellers to buyers.
They highlight the problem of data replication in combination with Shapley values and propose to overcome it heuristically by downweighing Shapley values according to data similarity.
Critically, their approach also reduces payoffs for honest, i.e., non-replicating parties.
In contrast, we achieve robustness-to-replication naturally through the multi-party setting in combination with novel characteristic functions.
Furthermore, in our marketplace no data is given to any party, avoiding further reselling of the data while in~\cite{agarwal2019marketplace} this is possible.
Other related work tries to price machine learning models instead of data directly~\cite{chen2018model}. In that work, a broker forms the interface between data sellers and model buyers.
Their main focus is to ensure affordability of models for all buyers by adjusting the model's performance through noise-injection and prevention of arbitrage.
\citeauthor{DBLP:conf/nips/AbernethyF11}~\shortcite{DBLP:conf/nips/AbernethyF11} study incentive mechanisms
in the setting where solving a task (e.g., ML prediction) is done in a collaborative decentralized manner with parties contributing towards solving the task with their expertise (e.g., improvements to a classifier), as opposed to contributing data. 
Marketplaces that treat data similar to classical goods are available at~\cite{bdex,qlik}.



\section{Future Work}

We consider the following extensions of our work:
\emph {Performance:} Improve the time to select training data using optimizations for evaluation of Shapley values such as sampling or influence functions~\cite{koh2017understanding,jia2019towards,ghorbani2019}.
\emph{Data Privacy:}  Use differential privacy techniques to strengthen privacy of training data in the customized models,
as a result mitigating model inversion or membership inference attacks on the output model~\cite{shokri2017membership,fredrikson2015model}.
\emph{Iterative Market:} Consider marketplace properties when the parties interact with the market over multiple rounds. Iterative markets would allow parties to learn additional information about the market such as the type of validation tasks and training data available, which might give advantage to adversarial parties.


\bibliographystyle{aaai}
\bibliography{refs}

  \clearpage
  \onecolumn
  \appendix


\section{Proof Details}

We show that our market instantiations are robust to replication.
Recall that a market is robust if
there is no party that gains a higher payoff in the market with its replicas
compared to its gain in the original market (Condition~\ref{cond:cmpshap}).
Let $i$ be the party that replicates itself.
Our proof proceeds as follows.
In Lemma~\ref{lemma:shapR}, we use
the Shapley value of party $i$ before replication ($\shap_i$)
to bound
the Shapley value of party~$i$ after replication ($\shapR_i$)
and show that the bound holds under several assumptions on the characteristic function~$v$.
We then show that, if $\shapR$ is bounded as it is in Lemma~\ref{lemma:shapR},
then the market is robust to replication (Lemma~\ref{lemma:robust}).
Finally, we prove that the characteristic functions for the single validation task market (Equation~\ref{eq:charS})
and multiple validations task market (Equation~\ref{eq:charM}) and corresponding Shapley values satisfy
conditions of  Lemma~\ref{lemma:robust} and hence are both robust-to-replication
(Theorems~\ref{thm:robustsingle} and~\ref{thm:robustmulti}).

In this section we outline lemmas that hold for both market instantiations and
in the following sections show that assumptions of the lemma statements indeed hold
for the single- and multiple-validation-task markets.

\begin{lemma}
\label{lemma:shapR}
In a collaborative ML market with parties in set $\parties$,
let $\price_i$ be the gain that party $i$ obtains for its own task
by participating in the market.
Let $v(S,\parties)$ be a characteristic function of the value of parties in $S \subseteq \parties$
with the following properties:
(1) $v(S,\parties) = v(S,\partiesR)$, $\forall S\subseteq \parties$,
where $\partiesR$ are the same parties as in the original
market along with a replica of party $i\in\parties$,~$i'$. (2) 
$v(S \cup\{i\},\partiesR) - v(S,\partiesR) \le a_i$ where $i' \in S$ and $i'$ is a replica of~$i$.
(3) $v$ is supermodular.
Let $\shap$ and $\shapR$ be the Shapley values
of the original market and the market with~$i$'s replica, then:
\begin{eqnarray*}
\shapR(i)  \le \frac{\shap(i) v(\parties,\parties)}{2v(\partiesR)} + \frac{\price_i }{2v(\partiesR, \partiesR)} 
= \frac{\shap(i) v(\parties, \parties) + \price_i }{2v(\partiesR, \partiesR)} 
\end{eqnarray*}
\end{lemma}
\begin{proof}
Let us express $\shapR(i)$ in terms of $\shap$.
Compare the coalition subsets $S$ between the two values.
All the new subsets in $\shapR$ will contain the replica of $i$
and, by the assumption of the lemma,
for such subsets $v(S\cup i,\partiesR) - v(S,\partiesR) = a_i$.
The other subsets are the same as in $\shap$
except their value towards the overall Shapley value decreases.
Let $\shapR(i) = \shapRold(i) + \shapRnew(i)$ where $\shapRold$ is the value brought from old (coalition)
subsets and $\shapRnew$ from the new ones.

\paragraph{Bounding $\shapRold(i)$:}
Since $v(S,\parties) = v(S,\partiesR)$ for all subsets $S$ where replica $i'$ is not present (i.e., those
subsets used in computing $\shapRold$),
we omit the second argument of $v$ to simplify the notation.

Recall that the (normalized) Shapley values of the subsets that existed before replication for both markets are:
$$\shap(i) \!= \frac{1}{v(\parties)}  \sum_{S \subseteq \parties \setminus \{ i \}} \frac{|S|! (M-|S|-1)!}{M!}  \big( v(S\!\cup\!\{i\}) \!-\! v(S) \big)$$

$$\shapR_{\mathsf{old}}(i) \!=  \frac{1}{v(\partiesR)}\sum_{S \subseteq \parties \setminus \{ i \}} \frac{|S|! (M-|S|)!}{(M+1)!}  \big( v(S\!\cup\!\{i\}) \!-\! v(S) \big)$$
We need to show that
\begin{equation}
\shapR_{\mathsf{old}}(i) = \frac{1}{v(\partiesR)}\sum_{S \subseteq \parties \setminus \{ i \}} \frac{|S|! (M-|S|)!}{(M+1)!}  \big( v(S\!\cup\!\{i\}) \!-\! v(S) \big) \le \frac{\shap(i) v(\parties)}{2v(\partiesR)}
~\label{eq:P1}
\end{equation}
Simplifying the inequality and substituting $\shap(v)$:
\begin{eqnarray*}
\sum_{S \subseteq \parties \setminus \{ i \}} \frac{|S|! (M-|S|)!}{(M+1)!}  \big( v(S\!\cup\!\{i\}) \!-\! v(S) \big) &\le& \frac{\shap(i) v(\parties)}{2}\\
\sum_{S \subseteq \parties \setminus \{ i \}} \frac{|S|! (M-|S|)!}{(M+1)!}  \big( v(S\!\cup\!\{i\}) \!-\! v(S) \big) &\le& \frac{1}{v(\parties)}  \sum_{S \subseteq \parties \setminus \{ i \}} \frac{|S|! (M-|S|-1)!}{M!}  \big( v(S\!\cup\!\{i\}) \!-\! v(S) \big)  \frac{v(\parties)}{2} \\
\sum_{S \subseteq \parties \setminus \{ i \}} \frac{|S|! (M-|S|)!}{(M+1)!}  \big( v(S\!\cup\!\{i\}) \!-\! v(S) \big) &\le& \frac{1}{2}   \sum_{S \subseteq \parties \setminus \{ i \}} \frac{|S|! (M-|S|-1)!}{M!}  \big( v(S\!\cup\!\{i\}) \!-\! v(S) \big)\\
\sum_{S \subseteq \parties \setminus \{ i \}} \frac{|S|! (M-|S|)!}{(M+1)}  \big( v(S\!\cup\!\{i\}) \!-\! v(S) \big) &\le& \frac{1}{2}   \sum_{S \subseteq \parties \setminus \{ i \}} {|S|! (M-|S|-1)!}  \big( v(S\!\cup\!\{i\}) \!-\! v(S) \big)\\
\sum_{S \subseteq \parties \setminus \{ i \}} {|S|! (M-|S|+1)!(M-2|S| -1)}  \big( v(S\!\cup\!\{i\}) \!-\! v(S) \big)  &\le& 0
\end{eqnarray*}

Note that coefficients $M-2|S| -1$ are symmetric
for sets $|R| < \lfloor M/2\rfloor $ and $|Q| = M-|R|-1$: $M-2|R| -1 = - (M-2|Q| -1)$.
There are ${M-1 \choose |R|}$ sets of size $|R|$ and ${M-1 \choose M-|R|-1}$ of size $M- |R|-1$.
Note that ${M-1 \choose |R|} = {M-1 \choose M-|R|-1}$. Hence, the number
of sets is the same and we are left to define a one-to-one mapping,
s.t., for every set $R$ there is a set $Q$, $R \subset Q$.
If this mapping exists, then
we can use supermodularity property of $v$: $v(R\cup\{i\}) - v(R) \le v(Q\cup\{i\}) - v(Q)$,
and the inequality is satisfied.

To define the mapping, we represent sets as a bipartite graph $(X,Y)$ where
vertices in $X$ represent all sets of size $|R|$ and vertices in $Y$
represent all sets of size $|Q|$.
For every $R \subset Q$,
we add an edge between the vertices that correspond to these two sets.
Note that every vertex in $X$ and $Y$ is connected to ${M-|R|-1 \choose M-2|R|-1}$ vertices
(i.e., we are counting the number of sets of size $M-|R|-1$ that contain subset $R$
and where other elements are chosen from the remaining $M-|R|-1$ elements).

Hall's marriage theorem states that there is a matching where every vertex 
$x \in X$ is matched if and only if $|W| \le |N(W)|$ for every $W \subseteq X$ and
where $N(W) \subseteq Y$ denotes neighbors of vertices in $W$. 
We prove that this is the case by construction.
Take a vertex $x\in X$ and initialize $W=\{x\}$.
Note that $|N(W)| = {M-|R|-1 \choose M-2|R|-1}$.
We proceed by expanding $W$ with arbitrary vertices from $X$.
We note that the total number of edges between $W$ and $N(W)$
is always $|W| \times {M-|R|-1 \choose M-2|R|-1}$.
For every vertex $x'$ that is
added there are two possibilities: either $x'$ contains same elements as sets already in $W$, in
which case $N(W) = N(W \cup \{x'\})$, or $N(W) \subset N(W \cup \{x'\})$ as sets that cover the new elements
are added. Suppose we are at the stage
when $|W| = |N(W)|$. We prove that there is no $x' \in X$ s.t.
$|N(W)| = |N(W \cup \{x'\})|$ and, hence, $W$ cannot be expanded
without also expanding $N(W)$.
Note that if $|W| = |N(W)|$ each vertex in $N(W)$ must be connected
to  ${M-|R|-1 \choose M-2|R|-1}$ vertices in $W$ to satisfy the constraint
on the number of edges from a vertex in $Y$ to $X$.
Hence, there cannot be a vertex in $X$ that has edges to $N(W)$ as this would exceed
the per-vertex number of edges constraint.
Hence, $x'$ does not exist.
Since $|X| = |Y|$ and there is a matching where every vertex
in $X$ is matched according to Hall's theorem, there is a mapping between every vertex in $X$ and $Y$.

\paragraph{Computing $\shapRnew(i)$:}
Let us now show that $\shapRnew(i) = \frac{\price_i }{2v(\partiesR)}$,
where we denote the number of coalitions of
size $s$ without party $i$ in the market with $M$ parties using, $n_{s,M}$:
$$n_{s,M} = {M-1 \choose s}$$

Each new coalition in the new marketplace contains $i'$.
The number of such new subsets is 
\begin{eqnarray*}
n_{M,M+1} + \sum^{M-1}_{s=1} n_{s,M+1}  -  n_{s,M} = 1+ \sum^{M-1}_{s=1} n_{s,M+1}  -  n_{s,M} 
\end{eqnarray*}

Since marginal contribution of $i$ in the new coalitions will be
limited to $\price_i$ by the property of $v$,
the overall value that these subsets bring to $\shapR(i)$ is:
\begin{eqnarray}
\shapRnew = \price_i \frac{1}{M+1} \frac{1}{v(\partiesR)} \left(1+ \sum^{M-1}_{s=1} \frac{n_{s,M+1}  -  n_{s,M}}{n_{s,M+1}}\right) = \nonumber\\
\price_i \frac{1}{M+1} \frac{1}{v(\partiesR)} \left(1+\frac{M-1}{2}\right) = \price_i \frac{1}{M+1} \frac{1}{v(\partiesR)} \frac{M+1}{2}=
\frac{\price_i }{2v(\partiesR)} 
 \label{eq:P2}
\end{eqnarray}
(see Claim~\ref{claim:factexp} for the expansion of the subset  combinations).

Then substituting derivations~\ref{eq:P1} and \ref{eq:P2} into $\shapR(i)$, we obtain the statement of the lemma.
%
\end{proof}

\begin{corollary}
If $\price_i = 0$, then 
$\shapR(i)  \le \frac{\shap(i) v(\parties)}{2v(\parties)}  = \frac{\shap(i)}{2}$. 
\end{corollary}

\begin{lemma}
\label{lemma:robust}
In a collaborative ML market with parties in set $\parties$,
let $\price_i$ be the gain that party $i$ obtains for its own task
by participating in the market.
Let $v$ be the characteristic function s.t.~$v(\parties) \ge \sum_{j\in\parties} \price_j$ and $v(\partiesR) = v(\parties) + \price_i$,
where~$\partiesR$ consists of the same parties as $\parties$ along with a replica of $i$.
If $\shapR \le  \frac{\shap(i) v(\parties, \parties) + \price_i }{2v(\partiesR, \partiesR)}$
for every party $i$,
then the market is robust to replication as per Condition~\ref{cond:cmpshap}.
\end{lemma}
\begin{proof}
Let us prove the contrary that the market is not robust to replication and hence
Condition~\ref{cond:cmpshap}
is false, i.e., that
$$\shapR(i) > \frac{\shap(i)\totalval + \price_i}{2(\totalval+\price_i)}
$$
where $\totalval = \sum_{j\in\parties} \price_j$.

Substituting $\shapR$ in Condition~\ref{cond:cmpshap}
and using the property that $v(\partiesR) = v(\parties) + \price_i$, we obtain
\begin{eqnarray}
\shap(i)\frac{\totalval}{2(\totalval+\price_i)} + \frac{\price_i}{2(\totalval+\price_i)}&<& \shap(i)\frac{ v(\parties)}{2(v(\parties) + a_i)} + \frac{\price_i }{2(v(\parties) + \price_i)} \nonumber \\ 
\frac{\shap(i)\totalval + \price_i}{\totalval+\price_i}&<& \frac{\shap(i)v(\parties) + \price_i}{v(\parties) + a_i}  \label{eq:cmp1}
\end{eqnarray}

Let $c = v(\parties) - \totalval$, where $c\ge 0$ by the property of $v$. Then, substituting $v(\parties)$:
\begin{eqnarray*}
\frac{\shap(i)\totalval + \price_i}{\totalval+\price_i}&<& \frac{\shap(i)(\totalval+c) + \price_i}{\totalval + c + a_i}  \\
\frac{\totalval + c + a_i}{\totalval+\price_i} &<& \frac{\shap(i)(\totalval+c) + \price_i}{\shap(i)\totalval + \price_i} 
\end{eqnarray*}

Let $x = 1/\shap(i) \ge 1$, then
\begin{eqnarray*}
\frac{\totalval + c + a_i}{\totalval+\price_i} &<& \frac{\totalval+c + x\price_i}{\totalval + x\price_i} 
\end{eqnarray*}

Let $p = \frac{\totalval + c + a_i}{\totalval+\price_i}$, then $\totalval + c + \price_i = p (\totalval+\price_i)$.
Since $c \ge 0$, $\price_i \ge 0$, $p\ge 1$.
We need to show that
\begin{eqnarray*}
p < \frac{\totalval+c + x\price_i}{\totalval + x\price_i} = \frac{p(\totalval+\price_i) - \price_i + x\price_i}{\totalval + x\price_i}\\
p (\totalval + x\price_i)< {p(\totalval+\price_i) - \price_i + x\price_i}\\
px\price_i< {p\price_i - \price_i + x\price_i}\\
px < p + x -1
\end{eqnarray*}
However, $px \ge p + x -1$ for $p,x \ge 1$
and we arrive to contradiction.
\end{proof}

\begin{claim}
$ \sum^{M-1}_{s=1}  \frac{n_{s,M+1}  -  n_{s,M}}{n_{s,M+1}} = 
 \frac{1}{M} \frac{(M-1)M}{2} =\frac{M-1}{2}$.
 \label{claim:factexp}
\end{claim}
\begin{proof}
\begin{eqnarray*}
 \frac{(n_{s,M+1}  -  n_{s,M} )}{n_{s,M+1}} =  \frac{{M\choose s} - {M-1 \choose s}}{{M\choose s}} = 1 -  \frac{{M-1 \choose s}}{{M\choose s}} = 1 -  \frac{(M-1)!s!(M-s)!}{s!(M-s-1)!M!}
= 1 -  \frac{M-s}{M} = \frac{s}{M}
\end{eqnarray*}

and
\begin{eqnarray*}
 \sum^{M-1}_{s=1}  \frac{{M\choose s} - {M-1 \choose s}}{{M\choose s}}) =  \sum^{M-1}_{s=1} \frac{s}{M} = \frac{1}{M} \frac{(M-1)M}{2} =\frac{M-1}{2} 
\end{eqnarray*}
\end{proof}

\section{Single Validation Task: Proof Details}

\begin{lemma}
\label{lemma:supermS}
Let $v$ be the characteristic function as defined in Equation~\ref{eq:charS}.
Then if $\gain$ is monotonic and supermodular, $v$ is supermodular.
\end{lemma}
\begin{proof}
We want to show that for sets $R \subset Q, i \not\in Q$, $v(R\cup\{i\}) - v(R) \le v(Q\cup\{i\}) - v(Q)$.
Recall that $v(S)$ is defined as
$v(S) =
{\gain\big(\valdata; \model_{S})}+ \sum_{j \in S} \big[ {\gain\big(\valdata; \model_{S}) - \gain\big(\valdata; \model_j)} \big] $.
Then
$v(S\cup \{i\}) - v(S)$ is
\begin{eqnarray*}
\gain\big(\valdata; \model_{S\cup \{i\}})+ \sum_{j \in S\cup \{i\}} \big[ {\gain\big(\valdata; \model_{S\cup \{i\}}) - \gain\big(\valdata; \model_j)} \big]
-
{\gain\big(\valdata; \model_{S})}- \sum_{j \in S} \big[ {\gain\big(\valdata; \model_{S}) - \gain\big(\valdata; \model_j)} \big] = \\
(|S| + 1) [ \gain\big(\valdata; \model_{S\cup \{i\}})  - \gain\big(\valdata; \model_{S}) ]
+  \big[ {\gain\big(\valdata; \model_{S\cup \{i\}}) - \gain\big(\valdata; \model_i)} \big].
\end{eqnarray*}
%
Since $\gain$ is monotone, non-negative, supermodular and $R \subset Q$,
$$(|R| + 1) [ \gain\big(\valdata; \model_{R\cup \{i\}})  - \gain\big(\valdata; \model_{R}) ]
 \le 
(|Q| + 1) [ \gain\big(\valdata; \model_{Q\cup \{i\}})  - \gain\big(\valdata; \model_{Q}) ].$$
Because of monotonicity, this implies
$$(|R| + 1) [ \gain\big(\valdata; \model_{R\cup \{i\}})  - \gain\big(\valdata; \model_{R}) ]
 \le 
(|Q| + 1) [ \gain\big(\valdata; \model_{Q\cup \{i\}})  - \gain\big(\valdata; \model_{Q}) ] + \underbrace{\gain\big(\valdata; \model_{Q\cup \{i\}}) - \gain\big(\valdata; \model_{R\cup \{i\}})}_{\geq 0}.$$
Rearranging terms and comparing with the equation above yields the statement of the lemma.
\end{proof}

\begin{lemma}
\label{lemma:margcontrib}
Let $v$ be the characteristic function as defined in Equation~\ref{eq:charS}
where  $\gain$'s value does not change with replicated data,
and $i'$ be a replica party of $i$. Then
$v(S \cup \{i\})- v(S) \le \price_i$ for any set $S$ when $i' \in S$.
\end{lemma}
\begin{proof}
Let us expand $v(S \cup \{i\})$ using the property of the gain function that states that replication
does not change its value (i.e., $\gain(\valdata; \model_{S\cup \{i\}}) = \gain(\valdata; \model_{S})$):
\begin{eqnarray*}
v(S\cup \{i\}) = &&
{\gain\big(\valdata; \model_{S\cup \{i\}})}+ \sum_{j \in S\cup \{i\}} \big[ {\gain\big(\valdata; \model_{S\cup \{i\}}) - \gain\big(\valdata; \model_j)} \big] = \\
&&{\gain\big(\valdata; \model_{S})}+ \sum_{j \in S \cup\{i\}} \big[ {\gain\big(\valdata; \model_{S}) - \gain\big(\valdata; \model_j)} \big] 
\end{eqnarray*}
Then
\begin{eqnarray*}
v(S \cup \{i\})&-& v(S)  \nonumber \\
&=&  {\gain\big(\valdata; \model_{S})}+ \sum_{j \in S  \cup\{i\}} \big[ {\gain\big(\valdata; \model_{S}) - \gain\big(\valdata; \model_j)} \big] \\
& & - {\gain\big(\valdata; \model_{S})}- \sum_{j \in S} \big[ {\gain\big(\valdata; \model_{S}) - \gain\big(\valdata; \model_j)} \big] \\
&=& {\gain\big(\valdata; \model_{S}) - \gain\big(\valdata; \model_i)}\le \price_i 
\end{eqnarray*} 
\end{proof}

\thmrobustsingle
\begin{proof}
Single task market uses a characteristic function $v$ that satisfies properties of
Lemma~\ref{lemma:shapR}:
(1) $v(S,\parties) = v(S,\partiesR)$, $\forall S\subseteq \parties$,
since $v$ depends only on $S$ and not the second argument.
(2)  $v(S \cup\{i\},\partiesR) - v(S,\partiesR) \le a_i$ holds
as per Lemma~\ref{lemma:margcontrib} where second argument of $v$ is
omitted to simplify the notation.
(3) $v$ is supermodular as shown in Lemma~\ref{lemma:supermS}.
Hence, $\shapR$ is bounded as in Lemma~\ref{lemma:shapR}.
Furthermore, $\vsing$ satisfies properties of Lemma~\ref{lemma:robust}:
since $\vsing(S) =  {\gain\big(\valdata; \model_{S})} + \sum_{j \in S} \big[{\gain\big(\valdata; \model_S) - \gain\big(\valdata; \model_j)} \big]$
and $\price_i = \gain(\valdata, \model_{\parties}) - \gain(\valdata, \model_{i})$,
$v(\parties) =  {\gain\big(\valdata; \model_{S})} + \sum_{j\in\parties} \price_j \ge  \sum_{j\in\parties} \price_j$.
The second property, $v(\partiesR) = v(\parties) + \price_i$, holds due to the property of $\gain$ on replicated data:
$$\vsing(\partiesR) =  {\gain\big(\valdata; \model_{\partiesR})} + \sum_{j \in \partiesR} \big[{\gain\big(\valdata; \model_{\partiesR}) - \gain\big(\valdata; \model_j)} \big] =
{\gain\big(\valdata; \model_{\parties})} + \sum_{j \in \partiesR} \big[{\gain\big(\valdata; \model_\parties) - \gain\big(\valdata; \model_j)} \big] = \vsing(\parties) + \price_i.$$

Single task market instantiation satisfies properties of Lemma~\ref{lemma:robust}
on $v$ and $\shapR$ and, hence, is robust to replication.
\end{proof}

\section{Multiple Validation Tasks: Proof Details}

\begin{lemma}
\label{lemma:supermM}
Let $w$ be the characteristic function as defined in Equation~\ref{eq:charM}.
Then if $\gain$ is monotonic and supermodular, $w$ is supermodular.
\end{lemma}
\begin{proof}
We want to show that for sets $R \subset Q$, $w(R\cup\{i\}) - w(R) \le w(Q\cup\{i\}) - w(Q)$,
where the second argument is omitted since $w$ has the same dependency on $\parties$ for all
sets in the same market instantiation.

Recall that $w$ is defined as:
$$w(S\cup \{i\}) = {\sum_{j\in \partiesD} \gain\big(\valdata_j; \model^j_{S\cup\{i\}})}
+ \sum_{j \in S\cup\{i\}} \big[{\gain\big(\valdata_j; \model^j_{S\cup\{i\}}) - \gain\big(\valdata_j; \model^j_{j})} \big]$$
Then
$w(S\cup \{i\}) - w(S)$ is:
\begin{eqnarray*}
{\sum_{j\in \partiesD} \gain\big(\valdata_j; \model^j_{S\cup\{i\}})}
+ \sum_{j \in S\cup\{i\}} \big[{\gain\big(\valdata_j; \model^j_{S\cup\{i\}}) - \gain\big(\valdata_j; \model^j_{j})} \big]
-
{\sum_{i\in \partiesD} \gain\big(\valdata_i; \model^i_{S})}
- \sum_{i \in S} \big[{\gain\big(\valdata_i; \model^i_{S}) - \gain\big(\valdata_i; \model^i_{i})} \big]
 = \\
\sum_{j\in \partiesD}(\gain\big(\valdata_j; \model^j_{S\cup\{i\}})
- \gain\big(\valdata_i; \model^i_{S}))
+
 \sum_{j \in S} \big[\gain\big(\valdata_j; \model^j_{S\cup\{i\}})-
 \gain\big(\valdata_i; \model^i_{S})\big]
 +  
 \big[{\gain\big(\valdata_i; \model^i_{S\cup\{i\}}) - \gain\big(\valdata_i; \model^i_{i})} \big]
\end{eqnarray*}

Hence, $v(R\cup\{i\}) - v(R) \le v(Q\cup\{i\}) - v(Q)$ can be written as follows:
\begin{eqnarray*}
\sum_{j\in \partiesD}(\gain\big(\valdata_j; \model^j_{R\cup\{i\}})
- \gain\big(\valdata_i; \model^i_{R}))
+
 \sum_{j \in R} \big[\gain\big(\valdata_j; \model^j_{R\cup\{i\}})-
 \gain\big(\valdata_i; \model^i_{R})\big]
 +  
 \big[{\gain\big(\valdata_i; \model^i_{R\cup\{i\}}) } \big]
 \le \\
 \sum_{j\in \partiesD}(\gain\big(\valdata_j; \model^j_{Q\cup\{i\}})
- \gain\big(\valdata_i; \model^i_{Q}))
+
 \sum_{j \in Q} \big[\gain\big(\valdata_j; \model^j_{Q\cup\{i\}})-
 \gain\big(\valdata_i; \model^i_{Q})\big]
 +  
 \big[{\gain\big(\valdata_i; \model^i_{Q\cup\{i\}}) } \big]
\end{eqnarray*}
Since $\gain$ is monotonic:
\begin{eqnarray*}
\sum_{j\in \partiesD}(\gain\big(\valdata_j; \model^j_{R\cup\{i\}})
- \gain\big(\valdata_i; \model^i_{R}))
+
 \sum_{j \in R} \big[\gain\big(\valdata_j; \model^j_{R\cup\{i\}})-
 \gain\big(\valdata_i; \model^i_{R})\big]
 \le \\
 \sum_{j\in \partiesD}(\gain\big(\valdata_j; \model^j_{Q\cup\{i\}})
- \gain\big(\valdata_i; \model^i_{Q}))
+
 \sum_{j \in Q} \big[\gain\big(\valdata_j; \model^j_{Q\cup\{i\}})-
 \gain\big(\valdata_i; \model^i_{Q})\big]
\end{eqnarray*}
The inequality holds, as long as $\gain$ is supermodular.
\end{proof}

\begin{lemma}
\label{lemma:margincontribM}
Let $w$ be the characteristic function as defined in Equation~\ref{eq:charM}
where $\gain$'s value does not change with replicated data,
and $i'$ be a replica party of $i$. Then
$w(S \cup \{i\},\parties)- w(S,\parties) \le \price_i$ for any set $S\subset \parties$ when $i' \in S$.
\end{lemma}
\begin{proof}
We omit the second argument of $w$ to simplify the notation.
Consider the following values of $w$ where we use the property of the gain function that states that replication
does not change its value (i.e., $\gain(\valdata; \model_{S\cup \{i\}}) = \gain(\valdata; \model_{S})$):
  \begin{eqnarray*}w(S  \cup i) &=&  \sum_{j\in \partiesD} \gain\big(\valdata_j; \model^j_{(S,i)})
 + \sum_{j \in S,i} \big[{\gain\big(\valdata_j; \model^j_{(S,i)}) - \gain\big(\valdata_j; \model^j_{j})}\big]\\
  w(S  \cup i,i') &=&  \sum_{j\in \partiesD} \gain\big(\valdata_j; \model^j_{(S,i,i')})
 + \sum_{j \in S,i,i'} \big[{\gain\big(\valdata_j; \model^j_{(S,i,i')}) - \gain\big(\valdata_j; \model^j_{j})}\big]  \\
&=&   \sum_{j\in \partiesD} \gain\big(\valdata_j; \model^j_{(S,i)})
 + \sum_{j \in S,i,i'} \big[{\gain\big(\valdata_j; \model^j_{(S,i)}) - \gain\big(\valdata_j; \model^j_{j})}\big] 
  \end{eqnarray*}

Then
  \begin{eqnarray*}
  w(S  \cup i,i') - w(S  \cup i)  &=&  \\
& &   \sum_{j\in \partiesD} \gain\big(\valdata_j; \model^j_{(S,i)})
 + \sum_{j \in S,i,i'} \big[{\gain\big(\valdata_j; \model^j_{(S,i)}) - \gain\big(\valdata_j; \model^j_{j})}\big] \\
& &  -   \sum_{j\in \partiesD} \gain\big(\valdata_j; \model^j_{(S,i)})
 - \sum_{j \in S,i} \big[{\gain\big(\valdata_j; \model^j_{(S,i)}) - \gain\big(\valdata_j; \model^j_{j})}\big] = \\
& &  \big[{\gain\big(\valdata_i; \model^j_{(S,i)}) - \gain\big(\valdata_i; \model^j_{i})}\big] = \price_i
  \end{eqnarray*}
\end{proof}

\thmrobustmult
\begin{proof}
Multiple task market uses a characteristic function $w$ that satisfies properties of
Lemma~\ref{lemma:shapR}:
(1) $w(S,\parties) = w(S,\partiesR)$, $\forall S\subseteq \parties$,
since $w$ depends only on $S$ and distinct validation tasks that are the same for $\partiesR$ and $\parties$
since $\valdata_{i'} = \valdata_{i}$.
(2)  $w(S \cup\{i\},\partiesR) - w(S,\partiesR) \le a_i$ holds
as per Lemma~\ref{lemma:margincontribM}.
(3) $w$ is supermodular as shown in Lemma~\ref{lemma:supermM}.
Hence, $\shapR$ is bounded as in Lemma~\ref{lemma:shapR}.
Furthermore, $w$ satisfies properties of Lemma~\ref{lemma:robust}:
since
$$\vmult(S) ={\sum_{i\in \partiesD} \gain\big(\valdata_i; \model^i_{S})}
+ \sum_{i \in S} \big[{\gain\big(\valdata_i; \model^i_{S}) - \gain\big(\valdata_i; \model^i_{i})} \big]$$
and $\price_i ={\gain\big(\valdata_i; \model^i_{S}) - \gain\big(\valdata_i; \model^i_{i})}$,
$\vmult(\parties) =  {\sum_{i\in \partiesD} \gain\big(\valdata_i; \model^i_{S})} + \sum_{j\in\parties} \price_j \ge  \sum_{j\in\parties} \price_j$.

The second property, $\vmult(\partiesR) = \vmult(\parties) + \price_i$, holds due to the property of $\gain$ on replicated data:
\begin{eqnarray*}\vmult(\partiesR) = 
{\sum_{i\in \partiesD} \gain\big(\valdata_i; \model^i_{\partiesR})}
+ \sum_{i \in \partiesR} \big[{\gain\big(\valdata_i; \model^i_{\partiesR}) - \gain\big(\valdata_i; \model^i_{i})} \big]
=\\
{\sum_{i\in \partiesD} \gain\big(\valdata_i; \model^i_{\parties})}
+ \sum_{i \in \partiesR} \big[{\gain\big(\valdata_i; \model^i_{\parties}) - \gain\big(\valdata_i; \model^i_{i})} \big]
=
\vmult(\parties) + \price_i.
\end{eqnarray*}

Multiple task market instantiation satisfies properties of Lemma~\ref{lemma:robust}
on $w$ and $\shapR$ and, hence, is robust to replication.
\end{proof}

\end{document}